\documentclass[submission,copyright,creativecommons]{eptcs}
\providecommand{\event}{GandALF 2023} 

\usepackage{iftex}

\usepackage{subcaption}

  \usepackage[english]{babel}
  \usepackage[T1]{fontenc}        

\title{Counterfactual Causality for Reachability and Safety based on Distance Functions}
\author{Julie Parreaux
\institute{Aix Marseille Univ, CNRS, LIS, Marseille, France}
\email{julie.parreaux@univ-amu.fr}
\and
Jakob Piribauer
\institute{Technische Universit\"at Dresden, Germany \\ Technische Universit\"at M\"unchen, Germany}
\email{jakob.piribauer@tu-dresden.de}
\and
Christel Baier
\institute{Technische Universit\"at Dresden, Germany}
\email{christel.baier@tu-dresden.de}
}
\def\titlerunning{Counterfactual Causality for Reachability and Safety}
\def\authorrunning{J. Parreaux, J. Piribauer, and C. Baier}

\usepackage[dvipsnames]{xcolor}
\usepackage{etex}
\usepackage{wrapfig}
\usepackage{graphicx}
\usepackage{amssymb,amsmath,amsthm}

\usepackage[mathscr]{eucal}
\usepackage{verbatim}
\usepackage{comment}
\usepackage{mathtools}
\usepackage{amsfonts}
\usepackage{nicefrac}
\usepackage{color}
\usepackage{graphicx}
\usepackage{float}
\usepackage{tikz}
\usepackage{xspace} 
\usepackage{pgf}
\usetikzlibrary{arrows,automata}

\usepackage[usenames,dvipsnames]{pstricks}
\usepackage{epsfig}
\usepackage{pst-grad} 
\usepackage{pst-plot} 
\usepackage{pstricks-add}
\usepackage[linesnumbered,ruled]{algorithm2e}
\usepackage{hyperref}
\usepackage{cleveref}
\usepackage[backgroundcolor=magenta!10!white]{todonotes}
\usepackage{enumitem}

\usepackage{wasysym}
\usepackage{multirow}

\usetikzlibrary{automata,positioning,shadows}
\usetikzlibrary{calc, shapes, backgrounds,arrows}
\usetikzlibrary{chains, decorations.pathmorphing}
\usetikzlibrary{positioning,decorations.pathreplacing}
\usetikzlibrary{shapes.misc}

\tikzset{every loop/.style={looseness=7}, >=latex}
\tikzset{every picture/.style={>=latex}}
\tikzstyle{PlayerReach}=[draw,circle,minimum size=6mm,inner sep=1.5pt]
\tikzstyle{PlayerSafe}=[draw,rectangle,minimum size=6mm,inner sep=1.5pt]
\tikzstyle{Player}=[draw,diamond,minimum size=7mm,inner sep=1.5pt]
\tikzstyle{target}=[draw=none,rectangle, minimum size=1mm]
\tikzstyle{PlayerReachmin}=[draw,circle, minimum size=4mm,inner sep=0pt]
\tikzstyle{PlayerSafemin}=[draw,rectangle,minimum size=4mm,inner sep=0pt]
\tikzstyle{cause} =[draw=blue,line width = 0.5mm]

\tikzstyle{strat} =[minimum width=0.1cm,line width=0.01mm,draw=none]
\tikzstyle{vecArrow} = [decoration={markings,mark=at position
	1 with {\arrow[scale=.5,>=latex]{>}}}
]

\usepackage{thmtools, thm-restate}

\newtheorem{theorem}{Theorem}
\newtheorem{proposition}[theorem]{Proposition}

\theoremstyle{definition}
\newtheorem{example}[theorem]{Example}
\newtheorem{definition}[theorem]{Definition}
\newtheorem{remark}[theorem]{Remark}

\definecolor{darkgreen}{rgb}{0.1, 0.5, 0.1}
\definecolor{RZ}{HTML}{385803}
\colorlet{simonColor}{Yellow!30!white}
\colorlet{jakobColor}{Green!30!white}

\newcommand{\eqdef}{\,\stackrel{\mathclap{\tiny\mbox{def}}}{=}\,}
\newcommand{\markend}{\hfill ${ \lrcorner}$}

\newcommand{\T}{\mathcal{T}}

\newcommand{\cT}{\T}

\newcommand{\sinit}{s_{\mathit{init}}}
\newcommand{\wgt}{\mathit{wgt}}
\newcommand{\AP}{\mathsf{AP}}

\newcommand{\dpref}{d_{\mathit{pref}}}
\newcommand{\dhamm}{d_{\mathit{Hamm}}}
\newcommand{\dghamm}{d_{\mathit{gHamm}}}

\newcommand{\dlev}{d_{\mathit{Lev}}}
\newcommand{\dstrat}[1]{d^{#1}}

\newcommand{\dist}{\mathit{dist}}
\newcommand{\dHH}{\ensuremath{d^*}\xspace}

\newcommand{\game}{\ensuremath{\mathcal{G}}\xspace}
\newcommand{\Pl}{\ensuremath{\Pi}\xspace}
\newcommand{\ReachPl}{\ensuremath{\mathsf{Reach}}\xspace}
\newcommand{\SafePl}{\ensuremath{\mathsf{Safe}}\xspace}

\newcommand{\Locs}{\ensuremath{V}\xspace}
\newcommand{\LocsReach}{\ensuremath{\Locs_{\ReachPl}}\xspace}
\newcommand{\LocsSafe}{\ensuremath{\Locs_{\SafePl}}\xspace}

\newcommand{\LocsT}{\ensuremath{\Locs_\textsl{Eff}}\xspace}

\newcommand{\loc}{\ensuremath{v}\xspace}
\newcommand{\locinit}{\ensuremath{v_{\textsf{i}}}\xspace}
\newcommand{\locT}{\ensuremath{\loc_\textsl{Eff}}\xspace}

\newcommand{\Trans}{\ensuremath{\Delta}\xspace}
\newcommand{\trans}{\ensuremath{\delta}\xspace}

\newcommand{\loosestrategy}{\ensuremath{\sigma}\xspace}
\newcommand{\winstrategy}{\ensuremath{\tau}\xspace}
\newcommand{\strategy}{\ensuremath{\mu}\xspace}

\newcommand{\play}{\ensuremath{\xi}\xspace}
\newcommand{\looseplay}{\ensuremath{\pi}\xspace}
\newcommand{\winplay}{\ensuremath{\rho}\xspace}

\newcommand{\N}{\mathbb{N}}

\begin{document}
\maketitle

\let\thefootnote\relax\footnotetext{\textbf{Funding:}{ {This work was partly funded by DFG Grant 389792660 as part of TRR 248 (Foundations of Perspicuous Software Systems), the Cluster of Excellence EXC 2050/1 (CeTI, project ID 390696704, as part of Germany’s Excellence Strategy), and  the DFG projects BA-1679/11-1 and BA-1679/12-1, and the ANR project Ticktac (ANR-18-CE40-0015).}}}


\begin{abstract}
Investigations of causality in operational systems aim at providing human-understandable explanations of \emph{why} a system behaves as it does.
There is, in particular, a demand to explain what went wrong on a given counterexample execution that shows that a system does not satisfy a given specification. 
To this end, this paper investigates a notion of counterfactual causality in transition systems based on Stalnaker's and Lewis' semantics of counterfactuals in terms of most similar possible worlds and introduces a novel corresponding notion of counterfactual causality in two-player games. 
 Using distance functions between paths in transition systems to capture the similarity of executions, this notion defines whether reaching a certain set of states is a cause for the fact that a given execution of a system  satisfies an undesirable reachability or safety property. Similarly, using distance functions between memoryless strategies in reachability and safety games, 
 it is defined whether reaching a set of states is a cause for the fact that a given strategy for the player under investigation is losing.

The contribution of the paper is two-fold:
In transition systems, it is shown that counterfactual causality can be checked in polynomial time for three prominent distance functions between paths. 
In two-player games, 
the introduced notion of counterfactual causality is shown to be checkable in polynomial time for two natural distance functions between memoryless strategies.
Further,  a notion of  explanation that  can be extracted from a counterfactual cause and that pinpoints  changes to be made to the given strategy in order to transform it into a winning strategy  is defined. For the two distance functions under consideration, the problem to decide whether such an explanation imposes only minimal necessary changes to the given strategy with respect to the used distance function turns out to be coNP-complete and not to be solvable in polynomial time if P is not equal to NP, respectively.

\end{abstract}

\section{Introduction}

%
%
%

Modern software and hardware systems have reached a level of complexity that makes it impossible for humans to assess whether a system behaves as intended without tools  tailored for this task.
To tackle this problem, automated verification techniques have been developed. \emph{Model checking} is one prominent such technique: A model-checking algorithm takes  a mathematical model of the system under investigation and a formal specification of the intended behavior and  determines whether all possible executions of the model satisfy the specification.
While the results of a model-checking algorithm provide guarantees on the correctness of a system or affirm the presence of an error, 
their usefulness  is, nevertheless, limited as they do not provide a human-understandable explanation of the  behavior of the system.

To provide additional information on \emph{why} the system behaves as it does,
certificates witnessing the result of the model-checking procedure, in particular counterexample traces in case of a negative result, have been studied extensively (see, e.g., \cite{ClarkeGMZ95,MaPn95,CGP99,Namjoshi01}). Due to the potentially still enormous size of counterexample traces and other certificates, a line of research has emerged that tries
to distill  comprehensible explications of what {causes} the system to behave as it does using formalizations of   \emph{causality} (see, e.g., \cite{Pearl09,Peters2017,ICALP21}).

\paragraph*{Forward- and backward-looking causality}

There are two fundamentally different types of notions of causality: \emph{forward-looking} and \emph{backward-looking} notions \cite{vandePoel2011}.
In the context of operational system models, forward-looking causality describes general causal relations between events that might happen along some possible executions. Backward-looking causality, on the other hand, addresses the causal relation between events along a given execution of the system model. 
This distinction is captured in more general contexts by the distinction between
\emph{type-level} causality addressing general causal dependencies between events that might happen when looking forward in a world model, and \emph{token-level} or \emph{actual} causality, corresponding to the backward view, that addresses causes for a particular event that actually happened  (see, e.g., \cite{Halpern15}). 

Notions of \emph{necessary} causality are typically  forward-looking: A necessary cause $C$ for an effect $E$ is an event that occurs on every execution that exhibits the effect $E$ (see, e.g.,  \cite{Baier2022Causality}, and for a philosophical analysis of necessity in causes \cite{Mackie65}).
The backward view naturally arises when the task is to explain what went wrong after an undesired effect has been observed. In the verification context, the backward view is natural for explaining counterexamples, see e.g. \cite{Zeller02,BallNR03,GroveV03,RenieresR03,WYIG06,WangAKGSL2013}. Most of these techniques rely on the \emph{counterfactuality} principle, which has been originally studied in philosophy \cite{Hume1739,Hume1748,Stalnaker1968,Lewis1973,Lewis1973counterfactuals} 
and formalized mathematically by Halpern and Pearl \cite{HalpernP2001,HalpernP04,HalpernP05,Halpern15}. 
Intuitively, counterfactual causality requires that 
the {effect} would not have happened, if the {cause} had not occurred, in combination with some minimality constraints for causes.
The most prominent account for the semantics of the involved counterfactual implication is provided 
   by Stalnaker and Lewis 
\cite{Stalnaker1968,Lewis1973,Lewis1973counterfactuals} 
in terms of closest, i.e.,  most similar, possible worlds. 
The statement ``if the cause $C$ had not occurred, then the effect $E$ would not have occurred'' holds true if in the worlds that are most similar to the actual world 
and in which $C$ did not occur, $E$ also did not occur. 
Interpreting executions of a system as possible worlds,
the actual world is an execution $\pi$ where both the effect $E$ and its
counterfactual cause $C$ occur, while the effect $E$ does not occur 
in alternative executions that
are as similar as possible to $\pi$ and that do not exhibit~$C$.

%


For a more detailed discussion on the distinction between forward- and backward looking causality and related concepts for responsibility, we refer the reader, e.g., to \cite{vandePoel2011,YazdanpanahD16,YazdanpanahDJAL19,BaierFM21,ICALP21}.

%

%
%
%
%
%
%

\paragraph*{Defining counterfactual causality in transition systems and reachability games}

%

To define our back\-ward-looking notion of counterfactual causality in transition systems, we follow an approach similar to the one by Groce et al \cite{groce2006error} who presented a Stalnaker-Lewis-style formalization of  counterfactual dependence of events using distance functions.
We consider the case where effects are reachability or safety properties and causes are sets of states.
To illustrate the idea, let
 $\T$ be a transition system and let  $E$ and $C$  be disjoint sets of states of $\T$ indicating a reachability effect and a potential cause, respectively. 
Consider an execution $\pi$ that reaches the effect set and the potential cause set.
We  employ the counterfactual reading of causality by Stalnaker and Lewis by viewing executions as possible worlds  using a similarity metric  $d$ on paths:
Reaching $C$ was a cause for $\pi$ to reach $E$ if all paths $\zeta$, that do not reach $C$ and that are most similar to $\pi$ according to $d$ among all paths with this property, satisfy $\zeta\vDash \Box \neg E$, i.e.,  they do not reach $E$.
So, we first determine the minimal similarity-distance $d_{\min}=\min \{ d(\pi,\zeta) \mid \zeta\vDash \Box \neg C\}$ from $\pi$ to a path $\zeta$ that does not reach $C$.
Then, we check whether all paths that do not reach $C$ and have similarity-distance $d_{\min}$ to $\pi$ do not reach $E$: Do all 
$
\zeta\in \{\zeta^\prime\mid d(\pi,\zeta^\prime)=d_{\min}$  and $\zeta^\prime\vDash \Box \neg C\}$  satisfy $\Box \neg E$? 
If the answer is yes, it is the case that ``if $C$ had not occurred, then $E$ would not have occurred'' and so $C$ is a counterfactual cause for $E$ on $\pi$.

\begin{example}
\label{ex:traffic}
Consider the following distance function on paths in a labeled transition system $\cT$ with states $S$ and a labeling function $L\colon S\to \mathcal{A}$ for a set of labels $\mathcal{A}$:
For paths $\pi=s_0,s_1,\dots$ and $\pi^\prime=t_0,t_1,\dots$, we define 
$
\dist(\pi,\pi^\prime) = | \{n\in \mathbb{N} \mid L(s_n)\not=L(t_n) \}|$.
So, paths are more similar if their traces  differ at fewer positions. To determine whether $C$ is a cause for $E$ on $\pi$, we  first determine what the least number  $n_{\min}$ of changes to the state labels  of $\pi$ is  to obtain a path $\zeta$ that does not reach $C$. Then, we have to check whether all paths differing from $\pi$ in $n_{\min}$  labels and  not reaching  $C$ do  not reach $E$. If this is the case, $C$ is a counterfactual cause for $E$ on $\pi$ with respect to  $\dist$.

	\begin{figure}[t]
	\resizebox{.65\textwidth}{!}{
		\includegraphics{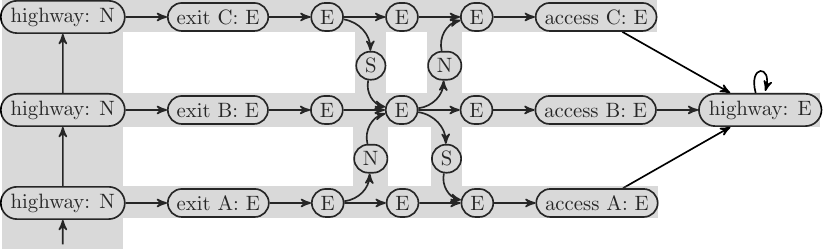}}	
		\hspace{12pt}			
		\resizebox{.3\textwidth}{!}{
		\begin{tikzpicture}
		[scale=1,->,>=stealth',auto ,node distance=0.5cm, thick]
		\tikzstyle{r}=[thin,draw=black,rectangle]
		
		\node[scale=1, rounded rectangle, draw] (start) {$a$};
		\draw[<-] (start) --++(0,+0.7);
		\node[scale=1, rounded rectangle, draw, below=.7 of start,xshift=-2cm] (step0) {$b$};
		\node[scale=1, rounded rectangle, draw, below=.7 of start,xshift=2cm] (step1) {$b$};
		
		\node[scale=1, rounded rectangle, draw, below=.7 of step0,xshift=-1cm] (step00) {$c$};
		\node[scale=1, rounded rectangle, draw, below=.7 of step0,xshift=1cm] (step01) {$a$};
		
		\node[scale=1, rounded rectangle, draw, below=.7 of step00] (step000) {$d$};
		
		\node[scale=1, rounded rectangle, draw, below=.7 of step01] (step011) {$d$};

		\node[scale=1, rounded rectangle, draw, below=.7 of step1] (step10) {$c$};

		\node[scale=1, rounded rectangle, draw, below=.7 of step10] (step100) {$d$};
		

		\draw[color=blue,rounded corners] (step1)+(-.4,-.3)    rectangle  node[right=0,yshift=.45cm] {$\mathit{cause}$}  ($(step1)+(+0.4,+0.3)$);
		
			\draw[color=red,rounded corners] (step011)+(-.4,-.3)    rectangle  node[above=0.25] {$\mathit{effect}$}  ($(step100)+(+0.4,+0.3)$);

	\draw[color=black,->] (start) edge  (step0) ;
	\draw[color=black,->,line width=2pt] (start) edge  (step1) ;
	\draw[color=black,->] (step0) edge  (step00) ;
	\draw[color=black,->] (step0) edge  (step01) ;
	\draw[color=black,->,line width=2pt] (step1) edge  (step10) ;
	\draw[color=black,->] (step00) edge  (step000) ;
	
	\draw[color=black,->] (step01) edge  (step011) ;
	
	\draw[color=black,->,line width=2pt] (step10) edge  (step100) ;
	

%
%
%
%
%

	\end{tikzpicture}}
\caption{On the left: example transition system modelling a traffic grid (Ex. \ref{ex:traffic}).  On the right: Example of a $\dhamm$-counterfactual cause that is not a $\dpref$-counterfactual cause (Ex. \ref{ex:dpref}).}
\label{fig:example_traffic}
		\end{figure}
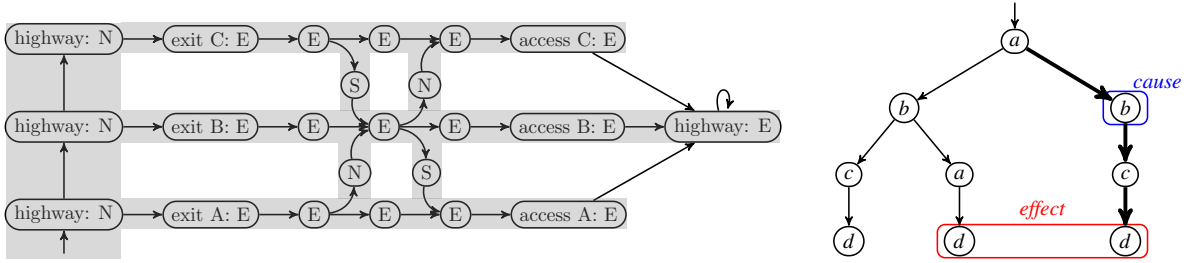

Now, consider the example transition system $\cT$ modelling a  road system with a highway going north that has three  exits into a small town which can be left again on a highway heading east depicted in Fig. \ref{fig:example_traffic}.
		Each state is labeled with $N$, $E$, or $S$ for north, east, and south as indicated in Fig. \ref{fig:example_traffic} depending on the direction the cars move on the respective road.
		Say, an agent  traverses the system via the path $\pi$ with trace $NNE^\omega$, i.e., by taking exit B from the first highway and
then going eastwards straight through the town.		
Assume that there is a traffic jam on access B while the other access roads are free. The question is now whether taking exit B was a cause for being stuck in slow traffic  later on, i.e., for the effect $\{\text{access B}\}$.
		First, note here that the set $\{\text{exit B}\}$ is  not a forward-looking necessary cause for reaching $\{\text{access B}\}$. There are  paths through the  system that avoid $\{\text{exit B}\}$, but reach $\{\text{access B}\}$.
		
		However, given the fact that the agent traversed the town by going straight eastwards, it is reasonable to say that the agent would have reached a different access road if she had taken a different exit from the first highway.
		This is  reflected in the counterfactual definition using  $\dist$: There are two paths that do not reach exit B and whose traces differ from $\pi$ at only one position, namely the paths with trace $NE^\omega$ and $NNNE^\omega$. These paths do also not reach access B. So,  $\{\text{exit B}\}$
		is a counterfactual cause for $\{\text{access B}\}$ on $\pi$; if the agent had taken exit A or C, she would not have hit the low traffic flow at access B. 
%
%
\markend
\end{example}

In the context of two-player reachability games, causality has  been used as a tool to solve games \cite{BaierCFFJS21}. 
In our work, we  focus on explaining why a certain strategy does not allow the player to win. 
More precisely, in reachability games between players 
with a safety and the complementing reachability objective, respectively, we consider the situation where one of the players $\Pl$ has a winning strategy, but loses the game using a strategy $\sigma$. We introduce a notion of counterfactual causality that aims to provide insights into what is wrong with strategy $\sigma$ by
transferring the counterfactual definition using  distance functions $d$ on memoryless strategies. A set of states $C$ is said to be a $d$-counterfactual cause for the fact that $\sigma$ is losing if all memoryless strategies $\tau$, that make sure that $C$ is not reached and have minimal $d$-distance to $\sigma$ among all such strategies, are winning. Furthermore, we introduce \emph{counterfactual explanations} that specify minimally invasive changes of $\sigma$'s decisions required to turn $\sigma$ into a winning strategy.

\paragraph*{Contributions}
\begin{itemize}
\item
We show that $d$-counterfactual causal relationships in transition systems (defined as in \cite{groce2006error}) 
can be checked in polynomial time for the following three distance metrics $d$ (Sec. \ref{sub:check_TS}):
\begin{enumerate}
\item
the prefix distance: paths are more similar if their traces share a longer prefix.
\item
the Hamming distance that counts the positions at which  traces of paths differ.
\item
the Levenshtein distance that counts how many insertions, deletions, and substitutions are necessary to transform the trace of one path to the trace of another path.
\end{enumerate}
Furthermore, we show that the  notion of $d$-counterfactual causality for the Hamming distance is consistent with Halpern and Pearl's but-for causes \cite{HalpernP04,HalpernP05}.
%
\item
In reachability games, we provide a generalization of this notion using similarity metrics on memoryless deterministic strategies. 
We show that for the Hausdorff lifting of the prefix distance on paths to a distance function on memoryless deterministic strategies, the resulting notion can be checked in polynomial time (Sec. \ref{sec:counterfactual_games}).
\item
We introduce a notion of \emph{counterfactual explanation} that can be computed from a counterfactual cause (Sec. \ref{sec:explain-strat}).
An explanation specifies where a non-winning  strategy needs to be changed. Of particular interest are $D$-minimal explanations that enforce only minimal necessary changes with respect to a distance function $D$ on strategies. For two distance functions related to the Hamming distance, we show that checking whether an explanation is minimal is coNP-complete and not in P if P$\not=$NP, respectively.
\end{itemize}
An overview of the complexity results can be found in Table~\ref{table:contribution}.


\begin{table}[tbp]
	\centering
	\hspace{-24pt}
	\begin{subtable}[b]{.3\linewidth}
	\begin{tabular}{|c|c|}
		\hline
		distance $d$ & causality \\
		\hline
		\multirow{2}{*}{prefix} & in P  \\
		& (Thm.~\ref{thm:check_dpref_TS}) \\
		\hline
		\multirow{2}{*}{Hamming} & in P  \\ 
		& (Thm.~\ref{thm:checking_dhamm}) \\
		\hline
		\multirow{2}{*}{Levenshtein} & in P \\
		& (Thm.~\ref{thm:checking_Levenstein}) \\
		\hline
	\end{tabular}
	\end{subtable}
	\begin{subtable}[b]{.6\linewidth}
	\begin{tabular}{|c|c|c|}
		\hline
		distance $d$ & causality & explanations \\
		\hline
		Hausdorff lifting $\dpref^H$  & 
		in P & \\
		of the prefix distance & (Thm.~\ref{thm:checking-pref-game}) & \\
		\hline
		Hamming strategy & in P in acyclic games &
		coNP-complete  \\
		distance $\dhamm^s$ & (Thm.~\ref{thm:checking-hamm-game}) & (Cor. \ref{cor:check-hamm-game}) \\
		\hline
		Hausdorff-inspired  &  & not in P if P$\not=$NP \\
		distance \dHH & & (Cor. \ref{cor:check-hamm-game}) \\
		\hline
	\end{tabular}
	\end{subtable}
	\caption{Overview of the complexity results. On the left, the complexities of checking $d$-counterfactual 
		causality  in transition systems, and on the right, the complexities of checking $d$-counterfactual 
		causality and  $d$-minimality of explanations  in reachability games.}
	\label{table:contribution}
\end{table}

\paragraph*{Related work}

Ways to pinpoint the problematic steps in a counterexample trace by 
localizing errors  have widely  been studied  \cite{Zeller02,BallNR03,GroveV03,RenieresR03,WYIG06,WangAKGSL2013}.
For counterfactuality in transition systems, we follow the approach of \cite{groce2006error} with distance metrics. In contrast to the causes in this paper,  causes in \cite{groce2006error} are formulas in an expressive logic that can precisely talk about the valuation of variables after a certain number of steps. Further  \cite{groce2006error} is not concerned with checking causality, but with finding causes, which, due to the expressive type of causes, algorithmically boils down to finding executions avoiding the effect with a minimal distance to the given one.

Based on  counterfactuality, Halpern and Pearl \cite{HalpernP04,HalpernP05,Halpern15} provided an influential formalization of causality using structural equation models, which   has served as the basis for various notions of causality in the verification context  (see,e.g., \cite{BeerBCOT2012,Leiner-FischerL2013}).
A key ingredient is the notion of \emph{intervention} to provide a semantics for the counterfactual implication in Hume's definition of causality. 
 An intervention in a structural equation model sets a variable to a certain value by force, ignoring its dependencies on other variables, and evaluates the effects of this enforced change. In a sense, a minimal set of interventions to avoid a cause then leads to a most similar execution avoiding the cause. We will discuss the relations between our definition and the Halpern-Pearl definition in more detail in Section \ref{sec:HP-causality}.
 In \cite{BeerBCOT2012}, interventions are employed to counterexample traces in transition systems by allowing to flip  atomic propositions along a  trace. In contrast to our notion of counterfactual causes, this is tailored for complex linear time properties, but does not provide insights for reachability and safety. Furthermore, the flipping of atomic propositions can be seen as a change in the transition system while our definition  considers alternative executions without manipulating the system.
In~\cite{coenen2022temporal}, the Halpern-Pearl approach is applied to provide a counterfactual definition of 
causality in reactive systems. 
A distance partial order, namely the subset relation on sets of positions at which traces differ, is used to describe which interventions are acceptable as they constitute minimal changes  necessary to avoid the cause.  Checking causality is  shown to be decidable by a formulation as a hyperlogic model-checking problem.
Furthermore, notions of necessary and sufficient causes as sets of states in transition systems have been considered \cite{Baier2022Causality}. These do not rely on the counterfactuality principle and are of forward-looking nature.




We are not aware of formalisations of causality in game structures. The related concept of responsibility, has been investigated in 
multi-agent models \cite{YazdanpanahD16,YazdanpanahDJAL19}. Notions of forward and backward responsibility of players in multi-player
game structures with acyclic tree-like  arena have been studied \cite{BaierFM21}.

For a detailed overview of work on causality and related concepts in operational models, we refer the reader to the survey articles \cite{Chockler16,ICALP21}.

\section{Preliminaries}
\label{sec:prelim}

We briefly present notions we use and our notation. For details, see \cite{BK08,GradelTW-02}.

\noindent \textit{Transition systems.}
A transition system is a tuple $\cT=(S,\sinit,\to,L)$ where $S$ is a finite set of states, $\sinit\in S$ is an initial state, $\to\subseteq S\times S$ is a transition relation and $L\colon S\to 2^{\AP}$ is a labeling function where $\AP$ is a set of atomic propositions.
A path in a transition system is a finite or infinite sequence of states $s_0 s_1\dots$ such that $s_0=\sinit$ and, for all suitable indices $i$, there is a transition from $s_i$ to $s_{i+1}$, i.e., $(s_i,s_{i+1})\in \to$. Given a path $\pi=s_0 s_1\dots$, we denote its trace $L(s_0)L(s_1)\dots$ by $L(\pi)$.
If there are no outgoing transitions from a state, we call the state \emph{terminal}.

\noindent \textit{Computation tree logic (CTL).} The branching-time logic CTL consists of state formulas that are evaluated at states in a transition system  formed by $\Phi ::= \top \mid a \mid \Phi\land \Phi \mid \neg \Phi \mid \exists \varphi \mid \forall \varphi$ 
where $a\in \AP$ is an atomic proposition 
 and path formulas evaluated on paths formed by $\varphi::= \bigcirc \Phi \mid \Phi \mathrm{U} \Phi$. The semantics of the temporal operators in path formulas is as usual. We use the abbreviations $\lozenge \Phi$ for $\top \mathrm{U} \Phi$ and $\Box \Phi = \neg \lozenge\neg \Phi$ and also allow sets of states $T$ in the place of state formulas. The semantics of $\exists \varphi$ are that there exists a path starting in the state at which the formula is evaluated that satisfies $\varphi$; $\forall \varphi$ is defined dually to that as usual. Model checking of CTL-formulas can be done in polynomial time. For details, see \cite{BK08}.

\noindent \textit{Reachability games.}
	A \emph{reachability game} is a tuple $\game = (\Locs, \locinit, \Trans)$ where  
	$\Locs = \LocsReach \uplus \LocsSafe \uplus \LocsT$ is the set of 
	vertices shared between players \ReachPl and 
	\SafePl, and  some target vertices $\LocsT$ ($\textsl{Eff}$ for effect).
	$\locinit \in \Locs \setminus \LocsT$ is the initial vertex and 
	$\Trans \subseteq \Locs \times \Locs$ is the set of edges.
We denote by 
$\Trans(\loc)$ the set of edges from \loc.
W.l.o.g., we assume that target vertices are terminal states, 
i.e. for all vertices $\loc \in \LocsT$, $\Trans(\loc) = \emptyset$. A 
\emph{finite play} is a finite sequence of vertices 
$\looseplay = \loc_0\loc_1\cdots \loc_k\in \Locs^*$ such
that for all $0\leq i<k$, $(\loc_i,\loc_{i+1})\in \Trans$. A \emph{play} is either 
a finite play ending in a target vertex, or an infinite sequence of vertices 
such that every finite prefix is a finite play.
Transition systems can be viewed as one-player games.

A \emph{strategy} for \ReachPl in a reachability game \game is a mapping
$\loosestrategy \colon \Locs^* \LocsReach \to \Locs$. A play or finite
play $\looseplay = \loc_0\loc_1\cdots$ is a \emph{\loosestrategy-play} if
for all~$k$ with $\loc_k \in \LocsReach$, we have 
$\loosestrategy(\loc_0\cdots\loc_k) = \loc_{k+1}$. A strategy \loosestrategy is 
an \emph{MD-strategy} (for memoryless deterministic) if for all finite plays \play and $\play'$ with the same last 
vertex, we have that $\loosestrategy(\play)=\loosestrategy(\play')$. In this paper, we mainly use MD-strategies and write 
$\loosestrategy(\loc_k)$ instead of $\loosestrategy(\loc_0\cdots\loc_k)$ for MD-strategies $\loosestrategy$.
 Moreover, under 
a (partial) MD-strategy \loosestrategy, we define the 
\emph{reachability game under \loosestrategy}, denoted by $\game^{\loosestrategy} = 
(\Locs, \locinit, \Trans^{\loosestrategy})$, by removing edges not chosen by $\sigma$, i.e.,
$
\Trans^{\loosestrategy} = \Trans \setminus 
	\{ (\loc, \loc')\in \Trans \mid  \loc \in \LocsReach\text{ and $\sigma(\loc)$ is defined and $\sigma(\loc)\not=(\loc, \loc')$} \}
$.
When \loosestrategy is completely defined, $\game^{\loosestrategy}$ is a 
transition system. Finally, a strategy is \emph{winning} if all 
\loosestrategy-plays starting in \locinit  end in a target vertex. Analogous definitions 
apply to \SafePl.
	In reachability games, either \ReachPl or \SafePl wins with an MD-strategy. This winning strategy can be computed in 
	polynomial time (see, e.g., \cite{GradelTW-02}).

\noindent \textit{Distance function.}
A \emph{distance function} on a set $A$ is a function $d\colon A \times A \to \mathbb{R}_{\geq 0}\cup\{\infty\}$ such that 
  $d(x,x)=0$  for all $x \in A$ and
  $d(x,y)=d(y,x)$  for all $x,y \in A$.
It is called a \emph{pseudo-metric} if additionally
 $d(x,y)+d(y,z)\geq d(x,z)$  for all $x,y,z \in A$, and
 a \emph{metric} if further 
 $d(x,y)=0$ holds iff $x=y$  for all $x,y \in A$.

\section{Counterfactual causes in transition systems}
\label{sec:counterfactual_in_TS}

In this section, we introduce the backward-looking notion of counterfactual causes in transition systems using distance functions (Section \ref{sub:definition_TS}). Afterwards, we prove that the definition can be checked in polynomial time for three well-known distance functions (Section~\ref{sub:check_TS}). Finally, we illustrate similarities between our notion of counterfactual causality to the definition of causality by Halpern and Pearl (Sec \ref{sec:HP-causality}).
Proofs omitted here can be found in the extended version \cite{extended}.

\subsection{Definition}
\label{sub:definition_TS}
The effects we consider  are reachability or safety properties $\Phi=\lozenge E$ or $\Phi=\Box \neg E$ for a set of states $E$. As the behavior of the system after $E$ has been seen is not relevant for these properties, we assume that $E$ consists of terminal states.
\begin{definition}[$d$-counterfactual cause in transition systems]
	\label{def:conterfactual_TS}
Let $\cT$ be a transition system and let $d$ be a distance function on the set of maximal paths of $\cT$.
Let $E$ be a set of terminal  states and let $C$ be a set of states disjoint from $E$. Let $\Phi=\lozenge E$ or $\Phi=\Box \neg E$.
Given a maximal path $\pi$ that visits $C$ and satisfies $\Phi$ in $\cT$, we say that $C$ is a \emph{$d$-counterfactual cause for $\Phi$ on $\pi$} if 
\begin{enumerate}
\item there is a maximal path $\rho$ in $\cT$ that does not visit $C$, and
\item all maximal paths $\rho$ with $\rho\vDash \Box \neg C$ with minimal distance to $\pi$  do not satisfy $\Phi$.
In other words,  all maximal paths $\rho$ with $\rho\vDash \Box \neg C$ such that $d(\pi,\rho)\leq d(\pi,\rho^\prime)$ for all $\rho^\prime$ with
$\rho^\prime \vDash \Box \neg C$ satisfy $\rho\vDash \neg \Phi$. 
\end{enumerate}
\end{definition}

The choice of the similarity distance $d$ of course heavily influences the notion of $d$-counterfactual cause.
In this paper, we will instantiate the definition with three distance functions that are among the most prominent distance functions between traces (or words). 
An experimental investigation to clarify in which situations what kind of distance functions leads to a desirable notion of causality, however, remains as future work.

 \noindent\textit{Prefix metrics $\dpref^{\AP}$ and $\dpref$:} given two paths $\pi$ and $\rho$, let $n(\pi,\rho)$ be the length of the longest common prefix of their traces $L(\pi)$ and $L(\rho)$. Then, 
$
\dpref^{\AP}(\pi,\rho)\eqdef 2^{-n(\pi,\rho)}$.
We can also define the distance on paths instead of traces, which will be used later on: $\dpref(\pi,\rho)\eqdef 2^{-m(\pi,\rho)}$ where 
 $m(\pi,\rho)$ is the length of the longest common prefix of $\pi$ and $\rho$ as paths.
This can be seen as a special case of $\dpref^{\AP}$ if we assume that all states have a unique label.

The prefix metric measures similarity in a temporal way saying that executions are more similar if they initially agree 
for a longer period of time. If no further structure of the transition system or meaning of the labels is known, this distance function might be a 
reasonable  choice for counterfactual causality.

 \noindent\textit{Hamming distance $\dhamm$:} Given two words $w=w_0\dots w_n$ and $v=v_0\dots v_n$ of the same length,
 we define
$
\dhamm(w,v)  \eqdef   |\{0\leq i \leq n \mid w_i \not= v_i \} |$.
For two maximal paths $\pi$ and $\rho$ of the same length in a transition system $\cT$ with labeling function $L$, we define 
$\dhamm(\pi,\rho) \eqdef \dhamm(L(\pi),L(\rho))$. So, the distance between two paths is the Hamming distance of their traces. 

The Hamming distance seems to be a reasonable measure if a system naturally proceeds through different layers, e.g., if a counter is increased in each step. Then, traces are viewed to be more similar if they agree on more layers. The temporal order of these layers, however, does not play a role.

 \noindent\textit{Levenshtein distance $\dlev$ \cite{levenshtein1966binary}:}
Given two words $w=w_0\dots w_n$ and $v=v_0\dots v_m$, the Levenshtein distance is defined as the minimal number of editing operations needed to produce $v$ from $w$ where the allowed operations are insertion of a letter, deletion of a letter, and substitution of a letter by a different letter. 
Formally, we define $\dlev$ in terms of \emph{edit sequences}.
Let $\Sigma$ be an alphabet and $v,w\in \Sigma^\ast\cup\Sigma^\omega$ be two words over $\Sigma$.
The \emph{edit alphabet} for $\Sigma$ is defined as 
$
\Gamma \eqdef  (\Sigma\cup\{\varepsilon\})^2 \setminus \{(\varepsilon,\varepsilon)\}
$
where $\varepsilon$ is a fresh symbol.
An edit sequence for $v$ and $w$ is now a word $\gamma\in \Gamma^\ast\cup\Gamma^\omega$ such that the projection of $\gamma$ onto the first component results in $v$ when all $\varepsilon$s are removed and the projection of $\gamma$ onto the second component results in $w$ when all $\varepsilon$s are removed.
E.g., let $\Sigma=\{a,b,c\}$, $v=abbc$ and $w=accbc$. One  edit sequence is 
$
\gamma= (a,a)(b,c)(\varepsilon,c)(b,b)(c,c)$.
The weight  of an edit sequence  $\gamma=\gamma_1\gamma_2\dots$  is defined as
$
\wgt(\gamma)=| \{i \mid \gamma_i\not= (\sigma,\sigma)\text{ for all $\sigma \in \Sigma$}\}|$.
Then, for all words $v\in \Sigma^\ast\cup\Sigma^\omega$ and $w\in \Sigma^\ast\cup\Sigma^\omega$, we define
$
\dlev(v,w)=\min\{\wgt(\gamma)\mid \gamma \text{ is an edit sequence for $v$ and $w$}\}$.
%
%
Again, we obtain a pseudo-metric on paths  via the Levenshtein metric on traces.
%

The Levenshtein distance might be particularly useful if 
labels model actions that are taken. Two executions that are obtained by sequences of actions that only differ by 
inserting or leaving out some actions, but otherwise using the same actions, are considered to be similar in this case.

\begin{example}\label{ex:dpref}
 Let us  illustrate  counterfactual causality for the prefix metric $\dpref$ and the Hamming distance $\dhamm$.
%
Consider the  transition system depicted in Figure \ref{fig:example_traffic}. A path $\pi$ as indicated by the bold arrows on the right via the potential cause to the effect  has been taken:
This is not a $\dpref$-counterfactual cause on $\pi$: The most similar paths to $\pi$ that do not reach $\mathit{cause}$  are both paths that move to the left initially. As one of these paths reaches $\mathit{effect}$, the set $\mathit{cause}$ is not a $\dpref$-counterfactual cause for reaching $\mathit{effect}$. 

Considering the distance function $\dhamm$ with the labels of the states as in Figure \ref{fig:example_traffic}, we get a different result:
The trace of $\pi$ is $abcd$. The paths that avoid the potential cause have traces $abcd$ and $abad$, respectively. So, the most similar path avoiding $\mathit{cause}$ is the path on the left with trace $abcd$ that also avoids $\mathit{effect}$. So, $\mathit{cause}$  is a $\dhamm$-counterfactual cause on $\pi$ for $\lozenge\mathit{effect}$. Intuitively, this can be understood as saying if the system had avoided $\mathit{cause}$ but otherwise behaved (as similar as possible to) as it did in terms of the produced trace, the effect would not have occurred. 
In particular, if labels represent actions that have been chosen, this is  a reasonable reading of causality.
\markend
\end{example}


%
%
%

\subsection{Checking counterfactual causality in transition systems}
\label{sub:check_TS}

In this section, we provide algorithms to check $d$-counterfactual causality for the three distance functions $\dpref^{\AP}$, $\dhamm$, and $\dlev$.
For these algorithms, a maximal execution $\pi$ of the system has to be given. We assume that $\pi$ is a finite path ending in a terminal state.
The problem to find causes that are small or satisfy other desirable properties is not addressed in this paper and remains as future work.
We will briefly come back to this in the conclusions.

%
%

\paragraph*{Prefix distance.}
First, we consider  $\dpref^{\AP}$-counterfactual causality and hence $\dpref$-counterfactual causality as a special case.

\begin{restatable}{theorem}{checkprefAPTS}
\label{thm:check_dpref_TS}
Let $\cT=(S,\sinit,\to,L)$ be a transition system, $E$ a set of terminal states, $C$ a set of states disjoint from $E$, 
and $\Phi=\lozenge E$ or $\Phi=\Box \neg E$.
Let $\pi=s_0 \dots s_n$ be an execution reaching $C$ and satisfying $\Phi$.
It is decidable in polynomial time whether $C$ is a $\dpref^{\AP}$-counterfactual cause for $\Phi$ on $\pi$.
\end{restatable}

\begin{proof}[Proof sketch]
The following algorithm solves the problem in polynomial time:
 First, we determine the last index $i$  s.t.  $C$ is not reached on any path with trace $L(s_0),\dots,L(s_i)$
and s.t. $C$ is avoidable from  some state that is reachable via a path with trace $L(s_0),\dots,L(s_i)$. In order to that,
we recursively construct sets $T_{j+1}$ of states that are reachable via paths with trace $L(s_0),\dots,L(s_{j+1})$
and check for all states $t\in T_{j+1}$ whether  $t\vDash \exists \Box \neg C$. If no such state exists, we have found the first index $j+1$ such that $C$ is not avoidable anymore after trace $L(s_0),\dots,L(s_{j+1})$; so  we have found $i=j$.
Now, we check whether $t \vDash \forall (\Phi \to \lozenge C)$ for all $t\in T_i$. If this is the case, $C$ is a $\dpref^{\AP}$-counterfactual cause for $E$ on $\pi$; otherwise, it is not.
\end{proof}


\paragraph*{Hamming distance.}

The Hamming distance is only defined for words of the same length. We will hence first consider only transition systems in which all maximal paths have the same length. We can think of such transition systems as being structured in layers with indices $1$ to $k$ for some $k$. Transitions can then only move from a state on layer $i<k$ to a state on layer $i+1$.
Afterwards, we consider a simple generalization of the Hamming distance to words of different lengths.

\noindent\textit{Original Hamming distance.}
Let $\cT=(S,\sinit, \to, L)$ be a transition system  in which all maximal paths have the same length $k$. We annotate all states with the layer they are on: For each state $s\in S$, there is a unique length $n\leq k$ of all paths from $\sinit$ to $s$. We will say that state $s$ lies on layer $n$ in this case.
By our assumption that effect states are terminal, the states $E$ are all located on the last layer $k$. We assume furthermore that all effect states have the same labels.

\begin{restatable}{theorem}{checkhammTS}
\label{thm:checking_dhamm}
Let $\cT=(S,\sinit, \to, L)$ be a transition system  in which all maximal paths have the same length $k$. Let $E$ be a set of terminal states and
let $C\subseteq S$ be a set of states disjoint from $E$. 
Let $\Phi=\lozenge E$ or $\Phi=\Box \neg E$.
Let $\pi=s_0 \dots s_n$ be an execution reaching $C$ and satisfying $\Phi$.
It is decidable in polynomial time whether $C$ is a $\dhamm$-counter\-factual cause for $\Phi$ on $\pi$.
\end{restatable}

\begin{proof}[Proof sketch]
We sketch the proof for the case that $\Phi=\lozenge E$. 
 We  equip the states in $S$ with a weight function $\wgt\colon S\to \{0,1\}$ such that the $\dhamm$-distance of a path to $\pi$ is equal to the accumulated weight of that path. 
 A state $t$ on layer $i$ gets weight $1$ if its label is different to $L(s_i)$. Otherwise, it gets weight $0$.
Now, we can check  whether $C$ is a $\dhamm$-counterfactual cause, as follows: We remove all states in $C$ and compute a shortest (i.e., weight-minimal) path $\zeta$ to $E$ and a shortest path $\xi$ to any terminal state. If the weight of $\xi$ is lower than the weight of $\zeta$, the paths avoiding $C$ that are $\dhamm$-closest to $\pi$ do not reach $E$ and $C$ is a $\dhamm$-counterfactual cause for $\lozenge E$ on $\pi$; otherwise, it is not.
\end{proof}

\begin{remark}\label{rem:weightedHamming}
The Hamming distance between paths could easily be extended to account for different levels of similarities between labels: Given a  similarity metric $d$ on the set of labels, one could define the distance between two paths $\pi=s_1\dots s_k$ and $\rho=t_1\dots  t_k$ as $\dhamm^\prime(\pi,\rho) \eqdef \sum_{i=1}^k d(s_i,t_i)$. 
 The algorithm  in the proof of Theorem \ref{thm:checking_dhamm} can now easily be adapted to this modified Hamming distance by defining the weight function on the transition system in the obvious way.
\end{remark}

\noindent\textit{Generalized Hamming distance.}
The assumption in the previous section that all paths in a transition system have the same length  is quite restrictive.
Hence, we now consider the following generalized version $\dghamm$ of the Hamming distance: For words $w=w_1\dots w_n$ and $v=v_1 \dots v_m$, we define
\[
\dghamm(w,v) \eqdef 
\begin{cases}
\dhamm(w,v_{[1:n]}) + (m{-}n) & \text{if $n\leq m$,} \\
\dhamm(w_{[1:m]},v) + (n{-}m)& \text{otherwise.} 
\end{cases}
\]
So $\dghamm$ takes a prefix of the longer word of the same length as the shorter word, computes the Hamming distance of the prefix and the shorter word, and adds the difference in length of the two words.

\begin{restatable}{theorem}{checkghammTS}
Let $\cT=(S,\sinit,\to,L)$ be a transition system, $E$ a set of terminal states, and $C$ a set of states disjoint from $E$.
Let $\Phi=\lozenge E$ or $\Phi=\Box \neg E$.
Let $\pi=s_0 \dots s_n$ be an execution reaching $C$ and satisfying $\Phi$.
It is decidable in polynomial time whether $C$ is a $\dghamm$-counter\-factual cause for $\Phi$ on $\pi$.
\end{restatable}

\begin{proof}[Proof sketch]
We adapt the proof of Theorem \ref{thm:checking_dhamm}:
We take $|\pi|$-many copies of the state space $S$ an let transitions lead from one copy to the next.
In the $i$th copy states with the same label as $s_i$ get weight $0$ and all other states get weight $1$.
Furthermore, we add transitions with weight $|\pi|-i$ from terminal states in a copy $i<|\pi|$ to the same state in the last copy to account for path that are shorter than $\pi$. 
The weight $|\pi|-i$ corresponds to the value added in the generalized Hamming distance when paths of different length are compared.
To account for paths longer than $\pi$, we furthermore allow transitions with weight $1$ within the last copy. These transitions are then taken until a terminal state is reached.
With these adaptations, the proof can be carried out analogously to  the proof of Theorem \ref{thm:checking_dhamm}.
\end{proof}

\paragraph*{Levenshtein distance.}
 The idea to check $\dlev$-counterfactual causality is to construct a weighted transition system to check causality via the computation of shortest paths as for the Hamming distance. 
So, let $\cT=(S,\to,\sinit,L)$ be a transition system labeled by $L$ with symbols from   $\Sigma=2^{\AP}$.
Let $E$ be a set of terminal  states and $C$ a set of states disjoint from $E$.
Let $\Phi=\lozenge E$ or $\Phi=\Box \neg E$.
Let $\pi=s_1\dots s_n$ be a maximal path reaching $C$ and satisfying $\Phi$.
The  transition system we construct  contains transitions corresponding directly to the edit operations insertion, deletion and substitution. A path in the constructed transition system  then corresponds
to an edit sequence between the trace of $\pi$ and the trace of another path in $\cT$.
This construction  shares some similarities with the construction of Levenshtein automata \cite{schulz2002fast} that accept all words with a Levenshtein distance below a given constant $c$ from a fixed word $w$.

Now, we formally construct the new weighted transition system $\cT_{\dlev}^\pi$:
The state space of this transition system is $S\times \{1,\dots, n\}$ with the initial state $(\sinit,1)$. The labeling function  is not used.
In $\cT_{\dlev}^\pi$, we allow the following transitions  labelled with letters from the edit alphabet~$\Gamma$:
\begin{enumerate}
\item
 a transition from $(s,i)$ to $(t,i+1)$ labeled with $(L(s_{i+1}),L(t))$ for each $(s,t)\in \to$ and $i<n$, 
\item
 a transition from $(s,i)$ to $(t,i)$ labeled with $(\varepsilon,L(t))$ for each $(s,t)\in \to$ and $i\leq n$,
\item 
 a transition from $(s,i)$ to $(s,i+1)$ labeled with $(L(s_{i+1}), \varepsilon)$ for each $s\in S$ and $i<n$.
\end{enumerate}
Note that the terminal states in $\cT_{\dlev}^\pi$ are all contained in $S\times\{n\}$.
Any maximal path in $\cT_{\dlev}^\pi$ corresponds to a maximal path $\rho$ in $\cT$.
This path $\rho$ is obtained by moving from a state $s$ to a state $t$ in $\cT$ whenever a corresponding transition of type 1 or 2 is taken in $\cT_{\dlev}^\pi$.
Transitions of type 3 do not correspond to a step in $\cT$ and stay in the same state.

Furthermore, given a finite path $\tau$  in $\cT_{\dlev}^\pi$ and the corresponding path $\rho=t_1\dots t_k$  in $\cT$, the labels of the transitions of $\tau$ form an edit sequence
for the words $L(s_2)\dots L(s_n)$ and $L(t_2)\dots L(t_k)$. 
To see this, observe that, for each $i>1$, whenever the copy $S\times\{i\}$ is entered in $\cT_{\dlev}^\pi$, the label of the transition contains $L(s_i)$ in the first component; if a transition stays in a copy $S\times\{i\}$, the label contains $\varepsilon$ in the first component. So, the projection onto the first component of the labels of the transitions of $\tau$  is indeed $L(s_2)\dots L(s_n)$, potentially with $\varepsilon$s in between. In the second component, whenever a transition of type 1 or 2 is taken, the label is simply the label of the corresponding state in $\rho$. Transitions of type 3 have $\varepsilon$ in the second component of their label.
Note here that $\rho$ and $\pi$ both start in $\sinit$ and that we could hence add $(L(\sinit),L(\sinit))$ to the beginning of the  edit sequence to obtain an edit sequence for the full traces of $\tau$ and $\rho$.
Note that also for infinite paths $\tau=t_1t_2\dots$  in $\cT_{\dlev}^\pi$ the transition labels provide an edit sequence for the words
$L(s_2)\dots L(s_n)$ and $L(t_2)L(t_3)\dots$.
Vice versa,  a finite maximal path $\rho=t_1\dots t_k$  in $\cT$ together with an edit sequence $\gamma$ for $L(s_2)\dots L(s_n)$ and $L(t_2)\dots L(t_k)$  provides  a maximal path 
$\tau$ in 
$\cT_{\dlev}^\pi$: The occurrences of $\varepsilon$ in $\gamma$ dictate which type of transition to take while the path $\rho$ tells us which state to move to.
As $\gamma$ projected to the first component contains $L(s_2)\dots L(s_n)$ enriched with $\varepsilon$s exactly $n-1$ transitions of type 1 or 3 are taken in $\tau$ obtained in this way and we indeed reach the last copy $\cT\times\{n\}$. As $\rho$ ends in a terminal state $t_k$, we furthermore reach the terminal state $(t_k,n)$. 
Analogously, an infinite path $\rho$  in $\cT$ together with an edit sequence $\gamma$ for $L(\pi)$ and $L(\rho)$ yields an infinite path in $\cT_{\dlev}^\pi$.

Based on these observations, we equip $\cT_{\dlev}^\pi$ with a weight function $\wgt$ on transitions: Transitions labeled with $(\sigma,\sigma)$ for a $\sigma\in \Sigma$
get weight $0$, the remaining transitions get weight $1$.

\begin{restatable}{theorem}{checklevTS}
	\label{thm:checking_Levenstein}
Let $\cT=(S,\sinit,\to,L)$ be a transition system, $E$ a set of terminal states, and $C$ a set of states disjoint from $E$.
Let $\Phi=\lozenge E$ or $\Phi=\Box \neg E$.
Let $\pi=s_0 \dots s_n$ be an execution reaching $C$ and satisfying $\Phi$.
It is decidable in polynomial time whether $C$ is a $\dlev$-counterfactual cause for $\Phi$ on $\pi$.
\end{restatable}

\begin{proof}[Proof sketch]
With the construction of the weighted transition system $\cT_{\dlev}^\pi$ above, the check can be done via the computation of shortest paths as for the Hamming distance above.
\end{proof}

\subsection{Relation to Halpern-Pearl causality}
\label{sec:HP-causality}

In the sequel, we want to demonstrate how our definition of counterfactual causality relates to Halpern-Pearl-style definitions of causality in \emph{structural equation models} \cite{HalpernP04,HalpernP05,Halpern15}.
A structural equation model consists of variables $X_1,\dots, X_n$ with finite domains   that are governed by   equations 
$
X_{i}=f_i(X_1,\dots, X_{i-1}, C)
$
for all $i\leq n$. Here, $f_i$ is an arbitrary function for each $i$ and $C$ is an input parameter for the context. For our consideration, the context $C$ does not play a role and we will hence omit it in the sequel.
So,  the value of variable $X_{i}$ depends on the value of (some of) the variables with lower index and the dependency is captured by the function $f_i$.
 Halpern and Pearl  use  \emph{interventions}
to define causality for an effect $E$, which is a set of valuations of $X_1, \dots, X_n$. An intervention puts the value of a variable $X_{i}$ to some $\alpha$ that is different from $f_i(X_1,\dots,X_{i-1})$, i.e., disregarding the equation $f_i$. Afterwards, the subsequent variables are evaluated as usual or by further interventions. Halpern and Pearl define:
\begin{definition}
Let  $f_1,\dots, f_n$ over variables $X_1,\dots, X_n$  be a structural equation model as above and let $E$ be an effect set of valuations such that the valuation of $X_1,\dots, X_n$ obtained by the structural equation model belongs to $E$. A \emph{but-for-cause} is a minimal subset $X\subseteq\{X_1,\dots, X_n\}$ with the following property:
There are values $\alpha_x$ for $x\in X$ such that putting variables $x\in X$ to $\alpha_x$ by intervention leads to a valuation of $X_1,\dots,X_n$ not exhibiting the effect $E$. More precisely, letting $t_i$ be the valuation $[X_1=w_1,\dots,X_{i-1}=w_{i-1}]$, where
$
w_i =
 f_{i}(w_1,\dots,w_{i-1})$   
 if $X_i\not\in X$, and 
$w_i = \alpha_{X_i}$ if $X_i\in X$, 
we get that $t_{n+1}\not\in E$.
\end{definition}

In order to compare this to our notion of counterfactual causes, we view structural equation models as tree-like transition system $\cT$:
The nodes at level $i$ are valuations for the variables $X_1,\dots, X_{i-1}$. At each node $s$ at level $i$, two actions are available: The action $\mathsf{default}$ moves to the state on level $i+1$ where the valuation in $s$ is extended by setting $X_i$ to the value $f_i(X_1,\dots,X_{i-1})$ where the values for 
$X_1,\dots,X_{i-1}$ are taken from the valuation in $s$. The action $\mathsf{intervention}$ extends the valuation of $s$ by setting $X_i$ to any other value than the action $\mathsf{default}$. 
The labelling in $\cT$ assigns the label $\{\mathsf{intervention}\}$ to all states that are reached by the action $\mathsf{intervention}$. The remaining states  and the initial state with the empty valuation get the label $\emptyset$.
Given an effect $E$ as a set of valuations, we interpret this as the corresponding set of leaf states in $\cT$. 
The default path $\pi$ that always chooses the action $\mathsf{default}$ corresponds to evaluating the equations in the structural equation model without interventions. 
 We can now capture but-for-causality with $\dhamm$-counterfactual causality along the default path $\pi$ if all variables  are Boolean:

\begin{restatable}{proposition}{butforcausality}
\label{prop:butfor}
Let  $f_1,\dots, f_n$ over Boolean variables $X_1,\dots, X_n$  be a structural equation model and let $E$ be an effect set of valuations.
Let $X$ be a but-for-cause for $E$. Let 
$C_X$ be the set of all nodes in the transition system $\cT$ which are reached by a $\mathsf{default}$-transition for a variable $x\in X$.
Then, $C_X$ is a $\dhamm$-counterfactual cause for $E$ in $\cT$ on the default path $\pi$.
\end{restatable}

For non-Boolean variables, the definitions of but-for-causes and of $\dhamm$-counterfactual causes have one significant difference:
A but-for-cause $X$ merely requires the existence of values to assign to the variables in $X$ by $\mathsf{intervention}$ such that the effect is avoided. A $\dhamm$-counterfactual cause $C$ in $\cT$  requires that for all possible interventions on the variables in $X$, the effect is avoided. This universal quantification originates from the universal 
quantification over most similar worlds in the Stalnaker-Lewis semantics of counterfactual causality.

The minimality requirement of but-for-causes  does not have a counterpart in the  definition of $d$-counter\-factual causes.
This allows us to assert that a candidate set of states $C$ is a $d$-counterfactual cause for an effect even if it contains redundancies.
When trying to find $d$-counterfactual causes for a given effect, on the other hand, of course trying to find (cardinality-)minimal causes is a reasonable option.

Besides but-for causality, 
we can also capture actual causality as in \cite{Halpern15} in our framework in the case of Boolean structural equation models.
This is demonstrated in the extended version \cite{extended}.

\section{Counterfactual causality in reachability games}

The counterfactual notion of causality introduced and investigated in the previous section can be applied to reachability games $\game$: 
We take the perspective of a player $\Pi$. 
Given a strategy $\sigma$ for the 
opponent and a play in which $\Pi$ lost, we  apply the definition to the transition system obtained from $\sigma$ and $\game$ and the given play. 
This allows us to analyze whether avoiding a certain set of states while playing against strategy $\sigma$ as similarly as possible to the given play would have allowed $\Pi$ to win.
Depending on whether we take the perspective of $\ReachPl$ or $\SafePl$, the effect that the player loses the game is a safety or reachability property, which we considered as effects in transition systems.
 The need to be given a strategy for the opponent, however, constitutes a major restriction to the usefulness of this approach.
 All proofs omitted in this section can be found in the extended version \cite{extended}.

\subsection{$D$-counterfactual causality}
\label{sec:counterfactual_games}

We provide a definition of counterfactual causality in reachability games in the sequel in which we only need the strategy $\sigma$ with which the player $\Pi$  played and are interested in  why the strategy $\sigma$  allows the opponent to win the game.
Since both players have optimal MD-strategies in a reachability game, we restrict ourselves to 
MD-strategies in the definition. 

\begin{definition}
	\label{def:game_conterfactual-cause}
Let $\game$ be a reachability game with target set $\LocsT$. Let $\Pi$ be one of the two players  and let $\sigma$ be a MD-strategy for player $\Pi$. Let $C$ be a set of locations disjoint from $\LocsT$. Let $D$ be a distance function on MD-strategies.
We say that $C$ is a $D$-counterfactual cause for the fact that $\Pi$ loses using $\sigma$  if 
\begin{enumerate}
\item there are $\sigma$-plays that reach $C$ on which $\Pi$ loses,
\item there is an MD-strategy $\tau$ for player $\Pi$ that avoids $C$ (i.e., there is no $\tau$-play reaching $C$),
\item  all MD-strategies $\tau$ for player $\Pi$, that avoid $C$ and that have minimal $D$-distance to $\sigma$ among the strategies avoiding $C$, are winning for $\Pi$.
\end{enumerate} 
\end{definition}

If we take the perspective of player $\Pi$ in game $\game$ where the opponent $\bar\Pi$ does not control any locations, MD-strategies for $\Pi$ satisfying condition 1 of the definition are essentially simple paths satisfying a safety or reachability effect property (with additional information on the states that are not visited by the path). To some extent, the definition can now be seen as a generalization of the definition for transition systems for suitable distance functions $D$:
We say a strategy distance function $D$  \emph{generalizes a path distance function} $d$ if in games where $\bar\Pi$ does not control any location, for all strategies $\sigma, \tau$ for $\Pi$, we have  $D(\sigma,\tau) = d(\pi_{\sigma},\pi_\tau)$ where $\pi_\sigma$ and $\pi_\tau$ are the unique $\sigma$- and $\tau$-plays.
The definition that $C$ is a  $D$-counterfactual cause for $\sigma$ losing the game agrees with the definition that $C$ is a  $d$-counterfactual cause on $\pi_{\sigma}$ for 
$\lozenge \LocsT$ or $\Box \neg \LocsT$ 
 in acyclic games  in this case. In cyclic games, there is one caveat: The definition for games quantifies only over MD-strategies which induce a play that is a simple path or simple lasso. The definition for transition systems quantifies over more complicated paths as well.

\paragraph*{Hausdorff distance $\dpref^H$ based on the prefix metric $\dpref$.}
A way to obtain a strategy distance function generalizing a given path distance function is the use of the Hausdorff distance on the set of plays of the strategies \cite[Section~6.2.2]{DelfourZ-11}:
	Let \winstrategy and \loosestrategy be two MD-strategies, and 
	$d$ be a distance function over plays. 
	The Hausdorff distance $d^H$ based on $d$  is defined by 
\[	{d^H}(\loosestrategy, \winstrategy) = 
	\max \left\{ \sup_{ \text{ $\loosestrategy$-plays }\looseplay} \,
	\inf_{ \text{ $\winstrategy$-plays }\winplay} \,
	d(\looseplay, \winplay),  
	\sup_{ \text{  $\winstrategy$-plays }\winplay} \,
	 \inf_{ \text{  $\loosestrategy$-plays } \looseplay} \,
	d(\looseplay, \winplay) \right\}.
	\]
	Let us consider the Hausdorff distance $\dpref^H$ based on the prefix metric $\dpref$ assuming that all states have a unique label.
For two strategies $\sigma$ and $\tau$ for $\SafePl$, the distance $\dpref^H(\sigma,\tau)$ 
is $2^{-n}$ where $n$ is the least natural number such that there is a prefix of length $n$ of a $\tau$-play that is not a prefix of a $\sigma$-play, or vice versa.
In order to find strategies that are as similar as possible to a given strategy $\sigma$, we hence have to consider strategies that follow $\sigma$ for as many steps as possible.
This leads to an algorithm for checking $\dpref^H$-counterfactual causality in reachability games that shares some similarities with the algorithm for checking $\dpref^{\AP}$-counterfactual causality in transition systems. 

\begin{restatable}{theorem}{checkprefgames}
	\label{thm:checking-pref-game}
Let $\game=(V,\locinit,\Trans)$ where $V=\LocsReach\uplus\LocsSafe\uplus\LocsT$ be a reachability game with target set $\LocsT$ and $\sigma$ a MD-strategy for player $\Pi$. Let $C$ be a set of locations disjoint from $\LocsT$.
We can check in polynomial time whether $C$ is a $\dpref^H$-counterfactual cause for the fact that $\Pi$ is losing using $\sigma$  in $\game$.
\end{restatable}

For the Hausdorff lifting of $\dhamm$ or $\dlev$, the resulting notion of counterfactual causes in games is more complicated.
If we try to adapt the approach used in transition systems, we need a way to capture the minimum distance of a given strategy to the closest winning strategies. 
However,  shortest path games (as extension of the weighted transition systems used for $\dhamm$- and $\dlev$-counterfactual causes in transition systems) cannot be employed in an obvious way.
In this paper,  we now instead consider two further distance functions related to the Hamming distance for which we can provide algorithmic results.

\paragraph*{Hamming strategy distance.} 
Let \loosestrategy and \winstrategy be two MD-strategies for \Pl in  \game, we define the Hamming strategy distance  by $\dhamm^s(\loosestrategy, \winstrategy) = 
|\{\loc \in \Locs \mid \loosestrategy(\loc) \neq \winstrategy(\loc)\}|$. As the Hamming distance on paths counts positions at which traces differ, the Hamming strategy distance counts positions at which two MD-strategies differ. 
Using a similar proof using shortest-path games~\cite{KhachiyanBBEGRZ-07} as for Theorem~\ref{thm:checking_dhamm}, 
we obtain the following polynomial-time result in the case of \emph{aperiodic} games. 
\begin{restatable}{theorem}{checkhammgames}
	\label{thm:checking-hamm-game}
	Let $\game=(V,\locinit,\Trans)$  be an acyclic reachability game with target set $\LocsT$ and $\sigma$ a MD-strategy for player $\Pi$. Let $C$ be a set of locations disjoint from $\LocsT$.
	We can check in polynomial time whether $C$ is a $\dhamm^s$-counterfactual cause for the fact that $\Pi$ is losing using $\sigma$  in $\game$.
\end{restatable}

\paragraph*{Hausdorff-inspired distance \dHH.}
The distance function \dHH computes the number of vertices where two MD-strategies make 
distinct choices along each play of both MD-strategies. It  hence has some similarity to a Hausdorff-lifting of the Hamming distance on paths. This Hausdorff-lifting, however, counts the number of 
\emph{occurrences} of  vertices at which two paths differ (in their label).
Instead, for a play $\winplay=\loc_0\loc_1\dots$ and a strategy \loosestrategy for \Pl, we define 
the \emph{distance between 
	\loosestrategy and \winplay} $\dist(\winplay, \loosestrategy)$ as the number of vertices  
$\loc \in \Locs_{\Pl}$ (i.e., not the number of occurrences) such that there exists $i \in \N$ with $\loc = \loc_i$ in \winplay, and $\loosestrategy(\loc_i) \neq 
(\loc_i, \loc_{i+1})$.We define
\dHH for  two strategies $\winstrategy, \loosestrategy$ by 
\[
\dHH(\winstrategy, \loosestrategy) = \max \big(
\sup_{\winplay \mid \text{\winstrategy-play}} \dist(\winplay, \loosestrategy), 
\sup_{\looseplay \mid \text{\loosestrategy-play}} \dist(\looseplay, \winstrategy) 
\big).
\]
To simplify the notation, we define
$\dstrat{\winstrategy}(\loosestrategy) = 
\sup_{\winplay \mid \text{\winstrategy-play}} \dist(\winplay, \loosestrategy)$.
We prove that the threshold problem for
\dHH is NP-complete via a reduction from the \emph{longest simple path problem}:

\begin{restatable}{proposition}{propdHH}
	\label{prop:dHH_NP-c}
	Let \game be a reachability game, \loosestrategy, \winstrategy be two  
	MD-strategies for \Pl, and $k \in \N$ be a  threshold. Then deciding 
	if $\dHH(\winstrategy, \loosestrategy) \geq k$ is NP-complete.
\end{restatable}

The  proposition explains why understanding $\dHH$-counterfactual causes is complex. We leave a further investigation of such notions for future work. As a first step toward a better understanding, we turn our attention to 
a conceptually simpler notion, the explanation induced by a counterfactual cause.

\begin{example}
	Let us illustrate counterfactual causes according to distances on strategies. 
	We consider the reachability game depicted in the left of \figurename{~\ref{fig:game_strat-cause-2}} 
	and the non-winning strategy $\loosestrategy$ for \ReachPl depicted in~green. 
	Under $\dpref^H$ or $\dHH$, the counterfactual cause for \ReachPl is $\{v_2, v_3\}$. 
	Indeed, there exists one play that reaches $v_3$ and loses for \ReachPl, 
	and there exists a unique strategy that avoids $\{v_2, v_3\}$ by 
	changing the choice of $\loosestrategy$ in $v_1$. Moreover, this 
	counterfactual cause is minimal since $\{v_3\}$ is not a cause. 
	Indeed, the (losing) strategy that differs from $\loosestrategy$ 
	in v$_0$ and $v_1$ avoids $\{v_3\}$ with a minimal distance to $\loosestrategy$, 
	i.e.\ $2^{-2}$ for $\dpref^H$ and $1$ for $\dHH$. 
	Under $\dhamm^s$, the counterfactual cause for \ReachPl is $\{v_3\}$. 
	Indeed, two strategies exist with a distance of $1$ to
	$\loosestrategy$ according to the vertex where \ReachPl 
	changes its decision. In these two strategies, only one 
	avoids $\{v_3\}$: the strategy where \ReachPl change its decision in $v_1$.  
	\markend
\end{example}

\subsection{$D$-counterfactual explanation}
\label{sec:explain-strat}

Given a $D$-counterfactual cause, we want to explain what is wrong in the losing strategy for \Pl. In particular, we are interested in sets of locations $C$ such that \Pl could have won the game if she had not made the decisions of $\sigma$ in the locations in $C$. 

\begin{definition}
	\label{def:game_quasi-strat-cause}
	Let \game be a reachability game  and \loosestrategy be a non-winning 
	MD-strategy for \Pl. Let $E \subseteq \Locs_{\Pl}$.
	We call $E$ an \emph{explanation} in \game under \loosestrategy if there exists a winning 
	MD-strategy \winstrategy such that for 
	all vertices $\loc \in \Locs_{\Pl}$, $\winstrategy(\loc) = \loosestrategy(\loc)$ 
	iff $\loc \notin E$. We call such a $\tau$ an
	\emph{$E$-distinct \loosestrategy-strategy}.
\end{definition}

%
%

We note that the definition of an explanation does not refer to a distance function. 
However, given a $D$-counterfactual cause, we can  compute an explanation no matter which  distance $D$ is used. 

\begin{proposition}\label{prop:difference_explanation}
	Let $\game=(V,\locinit,\Trans)$ be a reachability game, $D$ a distance function on strategies and \loosestrategy be a non-winning 
	MD-strategy for \Pl. Let $C \subseteq \Locs$ be a $D$-counterfactual cause.
	We can compute an explanation $E$ (from $C$) in polynomial time.
\end{proposition}
\begin{proof}
	Let $\game'=(V\setminus C,\locinit,\Trans)$ be the reachability game. Since $D$ is a $D$-counterfactual cause, we 
	know that there exists a winning strategy $\tau$ in $\game'$. We can compute this strategy in  time polynomial in the size of $\game'$ with the attractor method and we define $E = \{v \mid \loosestrategy(v) \neq \tau(v) \}$.
\end{proof}

A winning strategy differing from $\sigma$ in $E$ might not have much in common with $\sigma$. 
For this reason, explanations that point out changes in the
decisions of $\sigma$ in $E$ that enforce only the minimal necessary change to obtain a winning strategy $\tau$ from $\sigma$ are of particular interest. 
We can use a distance function $D$ to quantify how much a strategy needs to be changed.

\begin{definition}
	Let \game be a reachability game  and \loosestrategy be a non-winning 
	MD-strategy for \Pl. For a distance function $D$ for MD-strategies, we call a explanation 
	$E$ a \emph{$D$-minimal explanation}, if there exists a winning $E$-distinct 
	\loosestrategy-strategy  $\winstrategy$ with 
	$d(\winstrategy, \loosestrategy) = 
	\min \{d(\strategy, \loosestrategy)\mid {\text{$\strategy$ is a winning MD-strategy for $\Pl$}}\}$.
\end{definition}

For a strategy $\sigma$ and an explanation $E$, the distance $\dhamm^s(\sigma,\tau)$ for an $E$-distinct $\sigma$-strategy $\tau$ is precisely $|E|$. So, $\dhamm^s$-minimal explanations are cardinality-minimal explanations.

\begin{figure}[t]
\centering
	\begin{minipage}{.35\textwidth}
		\scalebox{0.75}{\begin{tikzpicture}
			[scale=1,->,>=stealth',auto ,node distance=0.5cm, thick]
			\tikzstyle{r}=[thin,draw=black,rectangle]
			
			\node[scale=1,rectangle, draw] (start) {};
			\draw[<-] (start) --++(0,+0.7);
			\node[scale=1, circle, draw, left=1.5 of start] (step0) {$v_0$};
			\node[scale=1, circle, draw, right=1.5 of start] (step1) {$v_1$};
			
			\node[scale=1, rectangle, draw, below=.7 of step0,xshift=-1cm] (step00) {};
			\node[scale=1, rectangle, draw, below=.7 of step0,xshift=1cm] (step01) {$v_2$};
			
			\node[scale=1, rectangle, draw, below=.7 of step00,xshift=-.5cm] (step000) {};
			\node[scale=1, rectangle, draw, below=.7 of step00,xshift=.5cm] (step001) {};
			
			\node[scale=1, rectangle, draw, below=.7 of step01,xshift=-.5cm] (step010) {};
			\node[scale=1, rectangle, draw=none, below=.7 of step01,xshift=.5cm] (step011) {$\locT$};

			\node[scale=1, rectangle, draw, below=.7 of step1,xshift=-1cm] (step10) {$v_3$};
			\node[scale=1, rectangle, draw, below=.7 of step1,xshift=1cm] (step11) {};

			\node[scale=1, rectangle, draw=none, below=.7 of step10,xshift=-.5cm] (step100) {$\locT$};
			\node[scale=1, rectangle, draw, below=.7 of step10,xshift=.5cm] (step101) {};
			
			\node[scale=1, rectangle, draw, below=.7 of step11,xshift=-.5cm] (step110) {};
			\node[scale=1, rectangle, draw, below=.7 of step11,xshift=.5cm] (step111) {};

			\draw[color=black,->] (start) edge  (step0) ;
			\draw[color=black,->] (start) edge  (step1) ;
			\draw[color=black,->, ForestGreen,line width=2pt] (step0) edge node[left] {\textcolor{ForestGreen}{$\loosestrategy$}} (step00) ;
			\draw[color=black,->] (step0) edge  (step01) ;
			\draw[color=black,->, ForestGreen,line width=2pt] (step1) edge node[left] {\textcolor{ForestGreen}{$\loosestrategy$}} (step10) ;
			\draw[color=black,->] (step1) edge  (step11) ;
			\draw[color=black,->] (step00) edge  (step001) ;
			\draw[color=black,->] (step00) edge  (step000) ;
			
			\draw[color=black,->] (step01) edge  (step011) ;
			\draw[color=black,->] (step01) edge  (step010) ;
			
			\draw[color=black,->] (step10) edge  (step101) ;
			\draw[color=black,->] (step10) edge  (step100) ;
			
			\draw[color=black,->] (step11) edge  (step111) ;
			\draw[color=black,->] (step11) edge  (step110) ;

			\end{tikzpicture}}
	\end{minipage}
	\hfill
	\begin{minipage}{.6\textwidth}
		\begin{tikzpicture}[xscale=.7,every node/.style={font=\footnotesize}, 
		every label/.style={font=\scriptsize}]
		\node[PlayerSafe] at (0,0) (s0) {$\loc_0$};
		\node[PlayerReach] at (-2,-1) (s1) {$\loc_1$};
		\node[PlayerReach] at (2, -1) (s2) {$\loc_2$};
		\node[target] at (0,-2) (e) {$\locT$};
		\node[strat] at (-2.1, 0) (st1) {};
		\node[strat] at (0, -1) (st2) {};
		
		\draw[->]
		(s0) edge (s1)
		(s0) edge (s2)
		(s1) edge[loop above,min distance=10mm] node[above] {$\trans_0$} (s1)
		(s1) edge node[below] {$\trans_1$} (e)
		(s2) edge node[above] {$\trans_2$} (s1)
		(s2) edge node[below, xshift=.05cm] {$\trans_4$} (e)
		;
		
		\draw[ForestGreen, ->, line width=0.5mm]
		(s1) edge[bend left=10] node[left] {\textcolor{ForestGreen}{$\loosestrategy$}} (st1)
		(s2) edge node[above] {\textcolor{ForestGreen}{$\loosestrategy$}}(st2)
		;
		
		\begin{scope}[xshift=8.5cm]
		\node[PlayerSafe] at (0,0) (s0) {$\loc_0$};
		\node[PlayerReach] at (-2,-1) (s1) {$\loc_1$};
		\node[PlayerReach] at (2, -1) (s2) {$\loc_2$};
		\node[target] at (0,-2) (e) {$\locT$};
		
		\node[strat] at (-1.9, -2.1) (st1) {};
		\node[strat] at (0, -1) (st2) {};
		\node[strat] at (-3, 0) (st3) {};
		\node[strat] at (-3, -.65) (st4) {};
		
		\node[PlayerReach] at (-4.5,0) (w0) {$w_0$};
		\node[PlayerReach] at (-4.5, -1) (w1) {$w_1$};
		\node[PlayerReach,label={30:$\locT$}] at (-4.5, -2) (w2) {$w_2$};
		
		\draw[->]
		(s0) edge (s1)
		(s0) edge (s2)
		(s1) edge[loop below,min distance=10mm] (s1)
		(s1) edge (e)
		(s2) edge (s1)
		(s2) edge  (e)
		(w0) edge (w1)	
		(w0) edge (s0)
		(w1) edge (w2)	
		(w1) edge (s0)	
		;
		
		\draw[RedOrange, ->, line width=0.5mm]
		(s1) edge[bend left=10] node[right] {\textcolor{RedOrange}{$\loosestrategy'$}} (st1)
		(s2) edge node[above] {\textcolor{RedOrange}{$\loosestrategy'$}}(st2)
		(w0) edge node[above] {\textcolor{RedOrange}{$\loosestrategy'$}} (st3)
		(w1) edge node[above] {\textcolor{RedOrange}{$\loosestrategy'$}}(st4)
		;
		\end{scope}
		
		\end{tikzpicture}
	\end{minipage}	
	\caption{On the left and in the middle, two reachability games with initial vertex $v_0$ and 
		strategy $\sigma$ for $\ReachPl$ (depicted in green). On the right, the reachability game 
		obtained by reduction of Corollary~\ref{cor:check-hamm-game} from the game depicted in the middle 
		with initial vertex $w_0$ and $\loosestrategy'$ be a non winning strategy for \ReachPl.}
	\label{fig:game_strat-cause-2}
\end{figure}
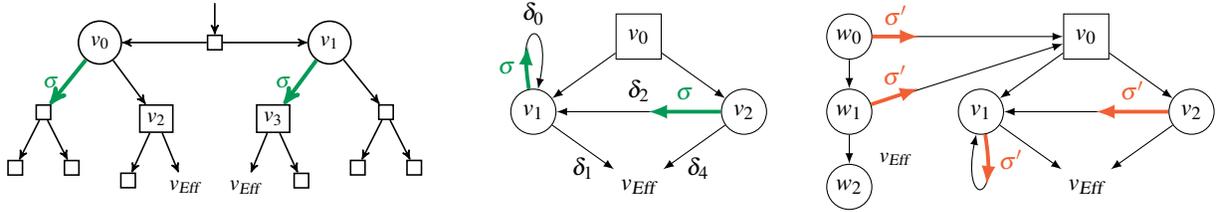

\begin{example} \label{ex:quasi-strat-cause}
	Let us illustrate explanations and $D$-minimal explanations.	
	We consider the reachability game \game where \ReachPl wins depicted 
	in the left of Figure {\ref{fig:game_strat-cause-2}} with \loosestrategy, a non-winning 
	MD-strategy for \ReachPl, depicted in green. 
	We note that $E = \{\loc_1, \loc_2\}$ 
	is an explanation in \game under \loosestrategy. A  winning 
	$E$-distinct \loosestrategy-strategy \winstrategy for \ReachPl is given by
	$\winstrategy(\loc_1) = \trans_1$ and $\winstrategy(\loc_2) = \trans_4$. 
	However, $E$ is not a $\dHH$-minimal explanation or $\dhamm^s$-minimal explanation. 
	Clearly, $\dhamm^s(\winstrategy,\loosestrategy)=2$. Further, also $\dHH(\winstrategy,\loosestrategy)=2$ as the $\sigma$-play 
	$\loc_0\loc_2\loc_1^{\omega}$ visits two states, namely $\loc_2$ and $\loc_1$ at which $\sigma$ and $\tau$ make different decisions.
	The set $E^\prime=\{\loc_1\}$, however is a $\dHH$-minimal explanation and $\dhamm^s$-minimal explanation:
	the $E^\prime$-distinct $\sigma$-strategy $\tau^\prime$ choosing $\trans_1$ in $\loc_1$ and behaving like $\sigma$ in $\loc_2$ wins and has $\dhamm^s$- and $\dHH$-distance $1$ to $\sigma$.
	As any winning strategy has at least distance $1$ to $\sigma$, $E^\prime$ is hence a $D$-minimal explanation for both distance functions.
	\markend
\end{example}

\label{subsec:strat_find-strat}

For $D$-minimal explanations, it is central to find a winning MD-strategy that minimises 
the distance $D$ to the given losing strategy $\sigma$. We take a look at this problem from the point of view of 
\ReachPl and prove that for $\dhamm^s$ and $\dHH$ the associated threshold problems are not in P if P$
\not=$NP.


\begin{restatable}{theorem}{stratFindingNPc}
	\label{thm:strat_finding-NP-c}
	Given a game $\game$,  a losing strategy $\sigma$ for $\ReachPl$, and $k\in \mathbb{N}$,
	deciding if there exists a winning MD-strategy \winstrategy for \ReachPl 
such that $\dhamm^s(\winstrategy, \loosestrategy) \leq k$ is  NP-complete. Further, the problem whether there is a winning MD-strategy \winstrategy with
$\dHH(\winstrategy, \loosestrategy) \leq k$ is not in P if P$\not=$NP.
\end{restatable}
\begin{proof}[Proof sketch]
	To establish the NP upper bound for $\dhamm^s$, we can guess a MD-strategy $\tau$ for \ReachPl and check in polynomial time whether 
	it is winning and whether $\dhamm^s(\winstrategy, \loosestrategy) \leq k$.
	For the NP-hardness for $\dhamm^s$, we provide a polynomial-time many-one reduction from the  NP-complete
	decision version of the \emph{feedback vertex set} 
	\cite{Karp1972}. Given a cyclic (directed) graph $G$, this problem asks whether there is a set $S$ of size at most $k$
	such that if we remove this set, $G \setminus S$ becomes acyclic. 
	For the problem for \dHH, we provide a polynomial-time Turing reduction from the same problem.
	 A detailed proof  is 
	given in \cite{extended}. 
%
\end{proof}

We deduce that   checking $D$-minimality of an explanation  
cannot be done in polynomial time  if P$\not=$NP.

\begin{restatable}{corollary}{checkdhammsstrat}
	\label{cor:check-hamm-game}
	Let \game be a reachability game, \loosestrategy be a 
	non-winning MD-strategy for \ReachPl,  and $E \subseteq \Locs$. The problem to check if $E$ is a $\dhamm^s$-minimal explanation in $\game$ for \loosestrategy
	is  coNP-complete. The problem to check if $E$ is a $\dHH$-minimal explanation in $\game$ for \loosestrategy
	is not in P if P$\not=$NP.
\end{restatable}

Despite the hardness in the general case, if $\game^{\loosestrategy}$ is acyclic, we prove 
that we can compute the winning MD-strategy that minimises the $\dHH$-distance to $\sigma$  in polynomial time.
 From this strategy, a $\dHH$-minimal explanation can then be computed as in Proposition \ref{prop:difference_explanation}.
 The  proof of 
	Theorem~\ref{thm:find-trans_poly}  (in the extended version \cite{extended})  constructs a shortest-path game~\cite{KhachiyanBBEGRZ-07} without negative weights in which an optimal strategy, that leads to the desired winning strategy in the original game, can be computed in polynomial time.

\begin{restatable}{theorem}{findTransPoly}
	\label{thm:find-trans_poly}
	Let \game be reachability game where \ReachPl wins, and \loosestrategy be a 
	non-winning MD-strategy for \ReachPl such that $\game^{\loosestrategy}$ is 
	acyclic. Then, we can compute a winning MD-strategy $\winstrategy$ that minimizes the distance \dHH to 
	\loosestrategy in polynomial time.
\end{restatable}
\section{Conclusion and Outlook}
\label{sec:conclusion}
The  notion of $d$-counterfactual cause for a distance function $d$ in transition systems turned out to be checkable in polynomial time for the distance functions $\dpref$, $\dhamm$, and $\dlev$ and so it has the potential to be employed in efficient tools to provide understandable explanations of the behavior of a system. 
In our algorithmic results for safety effects $\Phi$, one caveat remains: we  only considered finite executions reaching a cause candidate $C$ and satisfying $\Phi$. Allowing also finitely representable, e.g., ultimately periodic paths, constitutes a natural   extension, which requires adjustments in the provided algorithms.
 
The problem of finding good causes remains as future work: Whenever causality can be checked in polynomial time,  there is an obvious non-deterministic polynomial-time upper bound on the problem to decide whether there are causes below a given size, but the precise complexities are unclear. A further idea is to use the distance function to assess how good a cause is by considering  the distance from the actual execution to the closest executions avoiding a cause. For reachability effects and the prefix and Hamming distance, the set of direct predecessors optimizes this distance. For other distance functions or safety causes, this measure could, nevertheless, be more useful.
The search for similar  measures for the quality of causes constitutes an interesting direction for future  work.

In reachability games, we saw that the analogous definition of $D$-counterfactual causes can be checked in polynomial time for the Hausdorff-lifting $\dpref^H$ of the prefix metric, as well.
For other distance functions, 
the definition seems to lead to complicated notions due to the involved quantification over all MD-strategies avoiding the cause and having a minimal distance to a given strategy.
A closer investigation of these notions might, nevertheless, be a fruitful subject for future research.
However, our analysis of the 
 conceptually simpler $D$-minimal explanations provides insights into the complications one might encounter here.
 For the Hausdorff-inspired distance function $\dHH$, we showed that already the threshold problem for the distance between two given MD-strategies is NP-hard.
 Furthermore, for the relatively simple distance function $\dhamm^s$, checking the $\dhamm^s$-minimality of an explanation is in coNP-complete. For the Hausdorff-inspired distance function \dHH, checking \dHH-minimality is
  not in P unless P$=$NP.

\bibliographystyle{abbrv}

\bibliography{main}

\end{document}